\documentclass[conference,letterpaper]{IEEEtran}
\IEEEoverridecommandlockouts

\usepackage[utf8]{inputenc} 
\usepackage[T1]{fontenc}
\usepackage{url}
\usepackage{ifthen}
\usepackage{cite}
\usepackage[cmex10]{amsmath} 

\usepackage{mathtools}
\usepackage{amssymb}
\usepackage{xspace}
\usepackage{cleveref}
\usepackage{mleftright}

\usepackage{array}
\usepackage{siunitx,booktabs}
\usepackage{multirow}

\usepackage{amsthm}

\usepackage{tikz, pgfplots}
\pgfplotsset{compat=1.18}

\usepackage[dvipsnames]{xcolor}

\usepackage{balance}

\newcommand{\6}{\mathbf}

\newcommand{\getrand}{\xleftarrow{\smash{\raisebox{-0.4ex}{\tiny$\mathsf{R}$}}}}

\newcommand{\F}{\mathbb{F}_p}

\newcommand{\extdeg}{\mu}
\newcommand{\extdegVOLE}{\rho}
\newcommand{\K}{\mathbb{K}}

\newcommand{\KPolyXvec}[2]{\big(\K\mleft[X\mright]_{\leq#1}\big)^{#2}}

\newcommand{\com}{\mathsf{com}}
\newcommand{\openkey}{\mathsf{key}}

\newcommand{\openproof}{\mathsf{opening}_\evalpt}

\newcommand{\QPoly}[1]{Q\mleft(#1\mright)}
\newcommand{\QPolyX}{Q(X)}
\newcommand{\QPolyOnly}{Q}

\newcommand{\RelPoly}[1]{F\mleft(#1\mright)}
\newcommand{\RelPolyx}{F(\6x)}
\newcommand{\RelPolyOnly}{F}
\newcommand{\numc}{n} 
\newcommand{\witsize}{k} 

\newcommand{\RelPolyVec}[2]{\big(\F\mleft[x_1,\ldots,x_\witsize\mright]_{\leq#1}\big)^{#2}}

\newcommand{\evalpt}{r}
\newcommand{\domaineval}{\Omega}

\newcommand{\WitPoly}[1]{P_w\mleft(#1\mright)}
\newcommand{\WitPolyX}{P_w\mleft(X\mright)}
\newcommand{\WitPolyOnly}{P_w}

\newcommand{\OffPoly}[1]{\tilde{P}_w\mleft(#1\mright)}
\newcommand{\OffPolyX}{\tilde{P}_w\mleft(X\mright)}
\newcommand{\OffPolyOnly}{\tilde{P}_w}

\newcommand{\MaskPoly}[1]{P_u\mleft(#1\mright)}
\newcommand{\MaskPolyX}{P_u\mleft(X\mright)}
\newcommand{\MaskPolyOnly}{P_u}

\newcommand{\wit}{\6w}

\newcommand{\offset}{\boldsymbol{\Delta}_\wit}

\newcommand{\para}{\tau}

\newcommand{\size}[1]{\left\lvert#1\right\rvert} 

\newcommand{\cardinality}[1]{\left\lvert#1\right\rvert} 

\newcommand{\sign}{\sigma}

\newcommand{\signsym}{\sigma_\mathrm{sym}}

\newcommand{\signwit}{\sigma_{\wit}}

\newcommand{\signrel}{\sigma_{\RelPolyOnly}}

\newcommand{\logF}{\log_2 (\fieldsize)}

\newcommand{\Topen}{T_\mathrm{open}}

\newcommand{\nbatch}{\eta}

\newcommand{\BatchMat}{\boldsymbol{\Gamma}}

\newcommand{\seed}{\mathsf{seed}}
\newcommand{\GGMexp}[1]{\mathsf{VC}.\mathsf{exp}\mleft(#1\mright)}
\newcommand{\GGMopen}[1]{\mathsf{VC}.\mathsf{open}\mleft(#1\mright)}
\newcommand{\GGMpart}[1]{\mathsf{VC}.\mathsf{partial}\mleft(#1\mright)}
\newcommand{\GGMverify}[1]{\mathsf{VC}.\mathsf{verify}\mleft(#1\mright)}
\newcommand{\GGMsize}{N} 
\newcommand{\grindw}{w} 
\newcommand{\eset}{\mathcal{E}} 
\newcommand{\RSDP}{$\eset$-\textsf{SDP}\xspace}
\newcommand{\TernarySDP}{\textsf{Ternary-SDP}\xspace}
\newcommand{\CrossSDP}{\textsf{CROSS-SDP}\xspace}
\newcommand{\pcH}{\6H} 
\newcommand{\pcA}{\6A} 
\newcommand{\arow}{\6a} 
\newcommand{\Id}[1]{\61_{#1}}
\newcommand{\nC}{n} 
\newcommand{\kC}{k} 
\newcommand{\rC}{r} 
\newcommand{\synd}{\6s} 
\newcommand{\err}{\6e} 
\newcommand{\fieldsize}{p}
\newcommand{\sizeE}{z}

\newcommand{\checkE}[1]{f_{\eset}\mleft(#1\mright)}
\newcommand{\checkEx}{f_{\eset}(x)}
\newcommand{\checkEOnly}{f_{\eset}}

\newcommand{\inprod}[2]{\langle #1, #2\rangle}

\newcommand{\WAVE}{\textsf{WAVE}\xspace}
\newcommand{\CROSS}{\textsf{CROSS}\xspace}
\newcommand{\FAEST}{\textsf{FAEST}\xspace}
\newcommand{\SDitH}{\textsf{SDitH}\xspace}
\newcommand{\MQOM}{\textsf{MQOM}\xspace}
\newcommand{\LESS}{\textsf{LESS}\xspace}
\newcommand{\Mirath}{\textsf{Mirath}\xspace}
\newcommand{\PERK}{\textsf{PERK}\xspace}
\newcommand{\RYDE}{\textsf{RYDE}\xspace}
\newcommand{\SQISign}{\textsf{SQISign}\xspace}

\newcommand{\sounderr}{\varepsilon}
\newcommand{\EUFCMA}{\texttt{EUF-CMA}\xspace}

\newtheorem{theorem}{Theorem}

\newtheorem{remark}[theorem]{Remark}

\newtheorem{definition}[theorem]{Definition}
\newtheorem{property}{Property}
\newtheorem{proposition}{Proposition}

\newcommand{\logtwo}[1]{\log_2\mleft(#1\mright)}

\begin{document}
\title{TCitH- and VOLEitH-based Signatures\\ from Restricted
Decoding\thanks{{M.} {B.} is supported by the Italian Ministry of University and Research (MUR) under the PRIN 2022 program with projects ``Mathematical Primitives for Post Quantum Digital Signatures'' (CUP I53D23006580001) and ``Post quantum Identification and eNcryption primiTives: dEsign and Realization (POINTER)'' (CUP I53D23003670006), by MUR under the Italian Fund for Applied Science (FISA 2022), project ''Quantum-safe cryptographic tools for the protection of national data and information technology assets`` (QSAFEIT) - No. FISA 2022-00618 (CUP I33C24000520001). 
}}

\author{%
  \IEEEauthorblockN{\textbf{Sebastian Bitzer}$^1$, \textbf{Michele Battagliola}$^2$, \textbf{Antonia Wachter-Zeh}$^1$,   and \textbf{Violetta Weger}$^1$}
  \IEEEauthorblockA{$^1$Technical University of Munich, Germany, $^2$Università Politecnica delle Marche, Italy \\
                    \textit{\{sebastian.bitzer, antonia.wachter-zeh, violetta.weger\}@tum.de}, 
                    \textit{battagliola.michele@proton.me}
                    \vspace{-.5cm}}
}

\maketitle

\begin{abstract}

Threshold-Computation-in-the-Head (TCitH) and VOLE-in-the-Head (VOLEitH), two recent developments of the MPC-in-the-Head (MPCitH) paradigm, have significantly improved the performance of digital signature schemes.
This work embeds the restricted decoding problem within these frameworks:
we propose a structurally simple modeling that achieves competitive signature sizes.
Specifically, by instantiating the restricted decoding problem with the same hardness assumption underlying \CROSS, we reduce sizes by more than a factor of two compared to the NIST submission.
Moreover, we observe that ternary full-weight decoding, closely related to the hardness assumption underlying \WAVE, is a restricted decoding problem.
Using ternary full-weight decoding, we obtain signature sizes comparable to the smallest MPCitH-based candidates in the NIST competition.

\end{abstract}

\section{Introduction}
\textbf{Post-quantum signatures.} 
Following the standardization of the signatures Dilithium (ML-DSA), FALCON (FN-DSA), and SPHINCS\textsuperscript{+} (SLH-DSA), NIST issued a call for additional digital signature schemes in September 2022.
In October 2024, NIST announced that fourteen submissions proceeded to the second round.
Among these, six are based on the Threshold Computation in-the-Head (TCitH) and the Vector Oblivious Linear Evaluation (VOLE) in-the-head (VOLEitH) technique \cite{NISTPQC-ADD-R2:PERK24,NISTPQC-ADD-R2:MIRATH24,NISTPQC-ADD-R2:SDitH24,NISTPQC-ADD-R2:RYDE24,NISTPQC-ADD-R2:FAEST24,NISTPQC-ADD-R2:MQOM24}.
These recent specializations of the Multi-Party Computation (MPC) in-the-Head (MPCitH) paradigm offer significant performance improvements. 

\textbf{MPCitH.} 
Originally proposed in 2007 by Ishai, Kushilevitz, Ostrovsky, and Sahai \cite{STOC:IKOS07}, the MPCitH techniques have reshaped the landscape of post-quantum digital signatures.
In a nutshell, the MPCitH approach constructs digital signatures by combining secure MPC with the Fiat--Shamir transform.
Emulating an MPC protocol, the prover (i.e., the signer) performs computations on shares of a secret.
A zero-knowledge proof of knowledge is obtained by opening a verifier-selected subset of the parties, which is checked for consistency. 

\textbf{TCitH and VOLEitH.} The TCitH framework~\cite{feneuil2025threshold} instantiates the MPCitH paradigm with an $\ell$-out-of-$n$ threshold secret sharing. 
Alternatively, the TCitH framework may be described in the Polynomial Interactive Oracle Proofs (PIOP) formalism, as presented in~\cite{feneuil2024piop}.
This formalism also encompasses the VOLEitH framework, originally introduced as a technique to compile VOLE-based zero-knowledge protocols into publicly verifiable protocols, from which signatures can be obtained via the Fiat--Shamir transform
\cite{C:BBDKOR23}.

\textbf{Code-based Cryptography} refers to a class of cryptographic protocol whose security relies on hard problems in coding theory, such as syndrome decoding or code equivalence. 
In the context of the ongoing NIST competition, six out of the fourteen submissions are based on different coding-theory related assumption, such as code equivalence \cite{NISTPQC-ADD-R2:LESS24}, restricted decoding \cite{NISTPQC-ADD-R2:CROSS24}, Hamming-metric decoding \cite{NISTPQC-ADD-R2:SDitH24}, rank-metric decoding decoding~\cite{NISTPQC-ADD-R2:MIRATH24,NISTPQC-ADD-R2:RYDE24}, and the permuted kernel problem~\cite{NISTPQC-ADD-R2:PERK24}. 
Among these, the last four are based on either TCitH or VOLEitH.

\textbf{Our Contribution.}
In this work, we consider the restricted decoding problem, called \RSDP, underlying the \CROSS signature scheme~\cite{NISTPQC-ADD-R2:CROSS24}.
We present a modeling of \RSDP that allows constructing signatures in the TCitH and VOLEitH frameworks.
By doing so, we reduce \CROSS signature sizes by more than a factor of two, at the cost of a more complex construction.
Further, we consider the problem of ternary full-weight decoding, a special case of the hardness assumption underlying the \WAVE signature scheme~\cite{NISTPQC-ADD-R1:WAVE23,AC:DebSenTil19}.
Here, the proposed modeling achieves signature sizes comparable to the smallest\footnote{Except for \SQISign~\cite{NISTPQC-ADD-R2:SQIsign24}, which is considerably shorter and slower.} Zero-Knowledge-based (ZK-based) schemes of the NIST competition. 

\textbf{Organization.}
\Cref{sec:TCitH} provides a review of the TCitH and VOLEitH frameworks. \Cref{sec:crole} presents the proposed modeling of \RSDP, and \Cref{sec:performance} provides a detailed performance evaluation.
The repository~\cite{ourcode} contains code to reproduce the provided figures and parameters.

\section{TCitH and VOLEitH frameworks}

\begin{figure*}[t]
    \centering


\fbox{
\begin{tabular}{
p{0.6\columnwidth}
>{\centering}p{0.12\columnwidth}
p{0.65\columnwidth}
}
\multicolumn{1}{c}{\textbf{Prover}} &  & \multicolumn{1}{c}{\textbf{ Verifier}} \\[0.4em] 
$\seed \getrand \{0,1\}^\lambda$ &&\\
$\WitPolyOnly, \MaskPolyOnly,\com, \openkey \gets \GGMexp{\seed}$ &&\\
$\offset \gets \wit - \WitPoly{\infty}  \in \F^\witsize$ & $\xlongrightarrow{\com,\,\offset}$&\\
%
%
&&$\BatchMat  \getrand \K^{\nbatch \times \numc}$\\
$\OffPolyX \gets \WitPolyX + X\cdot \offset$&$\xlongleftarrow{\BatchMat}$&\\
$\QPolyX \gets \MaskPolyX + \BatchMat \cdot \RelPolyOnly(\OffPolyOnly)(X)$&$\xlongrightarrow{\QPolyOnly}$&\\
%
%
&&$\evalpt \getrand \domaineval \subseteq \K$\\
$\openproof \gets \GGMopen{\openkey,\evalpt}$  & $\xlongleftarrow{\evalpt}$  &   \\
&$\xlongrightarrow{\openproof}$&$\MaskPoly{\evalpt},\,\WitPoly{\evalpt} \gets \GGMpart{\openproof} $\\
&&$\OffPoly{\evalpt} \gets \evalpt\cdot\offset + \WitPoly{\evalpt}$\\
&&Check $\GGMverify{\com,\openproof,\evalpt}$\\ 
&&Check $\QPoly{\evalpt} = \MaskPoly{\evalpt}+ \BatchMat \cdot \RelPolyOnly(\OffPoly{\evalpt})$\\ 
\end{tabular}
}
    \caption{Five-pass zero-knowledge proof of knowledge.}
    \label{fig:5_pass_piop}
\end{figure*}

We provide a concise overview following the PIOP formalism presented in \cite{feneuil2024piop} used to build TCitH and VOLEitH signatures.

\textbf{Notation.}
$\lambda$ denotes the security parameter; NIST (security) categories 1, 3, and 5 correspond to $\lambda = 143$, $\lambda = 207$, and $\lambda = 272$, respectively.
For an integer $m$, we denote by $[m] =\{1,\ldots,m\}$.
Let $\F$ be the finite field with $\fieldsize$ elements, where $p$ is a prime number, and $\K$ its degree-$\extdeg$ field extension.
We denote by $\F[x_1, \ldots, x_\witsize]_{\leq d}$ the set of polynomials in variables $\6x = (x_1, \ldots, x_\witsize)$ with coefficients in $\F$ and total degree at most $d$.
Vectors are denoted with bold lowercase letters, and matrices with bold uppercase letters. 
Polynomials are denoted in plain lowercase letters, and vectors of polynomials in plain uppercase letters.
For a vector of polynomials in the variable $X$ $P(X)\in\F[X]_{\leq d}^\witsize$, we denote by $P(\infty)\in \F^\witsize$ the vector of leading (degree-$d$) coefficients. We write $x \propto y$ to denote that $x$ is proportional to $y$.

\textbf{TCitH.}
The general zero-knowledge proof of knowledge obtained from the TCitH framework is provided in \Cref{fig:5_pass_piop}.
We give a brief overview below; for further details, see \cite{feneuil2025threshold,NISTPQC-ADD-R2:MQOM24}. 
The prover attempts to convince the verifier that they know a witness $\wit \in \F^\witsize$ satisfying $\RelPoly{\wit} = \60\in\F^\numc$ for a One-Way Function (OWF) $\RelPolyOnly$, described by the system of polynomial relations $\RelPolyOnly = (f_1,\ldots, f_m) \in \RelPolyVec{d}{m}$.
To this end, the prover uses a GGM-based Vector Commitment (VC) scheme~\cite{AC:BBMORR24,C:BBDKOR23} to sample vectors of random polynomials 
\[\MaskPolyOnly \in \KPolyXvec{d-1}{\nbatch}
\quad \text{and}\quad \WitPolyOnly \in \KPolyXvec{1}{\witsize}
\]
with $\WitPoly{\infty} \in \F^\witsize$.
The VC also generates a commitment $\com$ to $\WitPolyOnly, \MaskPolyOnly$ and the opening key $\openkey$.
To embed the witness $\wit$ into $\WitPolyOnly$, the prover computes the offset $\offset = \wit - \WitPoly{\infty} \in \F^\witsize$, for which $\OffPolyX = \WitPolyX + \offset \cdot X$ has leading\footnote{The witness can also be embedded into the constant coefficient; see~\cite{NISTPQC-ADD-R2:MQOM24}.} coefficients $\wit$.
The offset $\offset$ and commitment $\com$ are then sent to the verifier. 
Next, the polynomial relation $\RelPolyOnly$ is evaluated at $\OffPolyX$.
Specifically, for $i\in[\numc]$ and $\ell \leq d$, denote by $f_{i,\ell}(\6x)$ the degree-$\ell$ homogeneous component (i.e. the sum of all the monomials of degree $\ell$) of $f_i(\6x)$ .
Then, $\RelPolyOnly(\OffPolyOnly)(X)$ is defined as
\[
\Big( \sum_\ell X^{d-\ell} f_{1,\ell}(\OffPolyX) ,
\ldots, \sum_\ell X^{d-\ell} f_{\numc,\ell}(\OffPolyX)\Big).
\]
Thus, $\RelPolyOnly(\OffPolyOnly)(\infty) = \RelPoly{\wit}$, i.e., the embedded witness satisfies the polynomial relation $\RelPolyOnly$ if and only if $\RelPolyOnly(\OffPolyOnly)(X)$ has degree at most $d-1$.
Hence, $\RelPolyOnly(\OffPolyOnly)(X)$ consists of $\numc$ polynomials of degree $d-1$.
The number of polynomials is reduced through batching.
To this end, the verifier samples a random matrix $\BatchMat \in \K^{\nbatch \times \numc}$ and the prover computes $\BatchMat \cdot \RelPolyOnly(\OffPolyOnly)$, thereby reducing the number of polynomials from $\numc$ to $\nbatch$ while preserving the degree with overwhelming probability.
The prover masks $\BatchMat\cdot\RelPolyOnly(\OffPolyOnly)(X)$ by computing
\[
\QPolyX =  \MaskPolyX+ \BatchMat\cdot\RelPolyOnly(\OffPolyOnly)(X) ,
\vspace{-1mm}\]
where $\MaskPolyX$ is the length-$\nbatch$ vector of random polynomials of degree $d-1$ generated by the VC scheme.
Sending the masked polynomial vector $\QPolyX$ to the verifier does not reveal any information about the witness.
To verify correctness, the verifier challenges the prover at a random evaluation point $\evalpt$ in the evaluation domain $\domaineval\subseteq \K$ of size $\GGMsize = \cardinality{\domaineval}$.
The prover produces the corresponding opening $\openproof$ using its opening key $\openkey$.
The verifier checks $\openproof$ against $\com$ and, in case of success, computes the evaluations $\WitPoly{\evalpt}\in\K^\nbatch$ and $\MaskPoly{\evalpt}\in\K^\witsize$.
Then, the verifier computes $\MaskPoly{\evalpt} + \BatchMat \cdot \RelPolyOnly(\OffPoly{\evalpt})$ and checks whether it matches $\QPoly{\evalpt}$, where $\QPolyX$ is the polynomial vector received from the prover.

\begin{remark}
In \cite{NISTPQC-ADD-R2:MQOM24}, also a three-pass variant of \Cref{fig:5_pass_piop} is given, which is obtained by skipping the batching step.
This simplifies the protocol, but increases the required communication, which is why we omit this variant.
\end{remark}

\textbf{TCitH signature scheme.}
The proof system above has a soundness error given by
\vspace{-1mm}
\[
\sounderr = \frac{1}{\fieldsize^{\extdeg\cdot \nbatch}} + \left(1-\frac{1}{\fieldsize^{\extdeg\cdot \nbatch}}\right)\frac{d}{\GGMsize}.
\]
\vspace{-1mm}%
For a target security of $\lambda$\,bits, $\para$ parallel repetitions are performed to amplify soundness below $2^{-\lambda}$.
The Fiat--Shamir transform (enhanced by a $\grindw$-bit proof of work, known as grinding) yields a signature scheme with the following properties.

\begin{property}
For grinding parameter $\grindw$, let the TCitH proof system parameters satisfy
\vspace{-2mm}
\[
\GGMsize \leq \fieldsize^\extdeg,
\qquad
\fieldsize^{\extdeg \cdot \nbatch} \geq 2^\lambda,
\quad\text{and}\quad
\left(\frac{\GGMsize}{d}\right)^\tau \geq 2^{\lambda-\grindw}.
\]
\vspace{-1mm}
Then the resulting signature scheme is \EUFCMA-secure~\cite{NISTPQC-ADD-R2:MQOM24}.

Signatures are of the form $\sign = (\signsym, \signwit,\signrel)$, where $\signwit$ corresponds to the witness, $\signrel$ to the OWF $\RelPolyOnly$, and $\signsym$ to the VC scheme.
The components' sizes\footnote{In general, the sizes depend on the specific instantiation of the VC scheme~\cite{EC:GYWZXZ23}, and optimizations not described here can be applied. The formulas given follow~\cite{NISTPQC-ADD-R2:MQOM24}, to which we refer the reader for further details.} are given by
$$
\size{\signsym}= \tau \cdot \lambda\cdot (\log_2 (N) +1) + 5\lambda + 32,
$$
$$
\size{\signwit} = \tau \cdot \witsize \cdot \log_2 (\fieldsize), \text{ } 
\size{\signrel} = \tau \cdot \nbatch \cdot (d-1)  \cdot \extdeg \cdot \log_2 (\fieldsize).\\
$$
\end{property}

\textbf{VOLEitH.}
Instead of performing $\para$ parallel repetitions, VOLEitH employs a technique that allows sampling $\evalpt$ from an evaluation domain $\domaineval$ of size 
\[
\cardinality{\domaineval} = \GGMsize^\para 
\approx 2^\lambda
\leq \cardinality{\K} = \fieldsize^{\extdegVOLE}.
\]
Avoiding parallel repetitions reduces the soundness error, which in turn enables shorter signatures.
However, this comes at the cost of an additional round of interaction between the prover and the verifier to do a a \emph{consistency check}.
For further details, we refer to \cite{NISTPQC-ADD-R2:FAEST24,C:BBDKOR23}.

By parameterizing the consistency check with $B$ and the state-of-the-art VC scheme~\cite{AC:BBMORR24} with $\Topen$, a signature scheme with the following properties can be constructed.

\begin{property}
Let the VOLEitH proof system be parameterized such that
\begin{align*}
\frac{\GGMsize^\tau}{d} \geq 2^{\lambda-\grindw}
\quad
\text{and}
\quad
\fieldsize^\rho \geq 2^\lambda, \\
\text{with}
\quad
\extdegVOLE = \para \cdot \left\lceil\log_{\fieldsize}( \GGMsize)\right\rceil,
\quad
B = \left\lceil\frac{16}{\logF}\right\rceil.
\end{align*}
Then the resulting signature scheme is \EUFCMA-secure~\cite{NISTPQC-ADD-R2:FAEST24}, and its signatures have size\footnote{The provided formulas follow the specific design choices made in~\cite{NISTPQC-ADD-R2:SDitH24}.}
$\size{\sign} = \size{\signsym}+ \size{\signwit} + \size{\signrel}$, where
\begin{align*}
\size{\signsym} &= \tau \cdot (\extdegVOLE + B) \logF + \tau \cdot 2\lambda + \Topen \cdot \lambda + 4\lambda + 32,\\
\size{\signwit} &= \tau \cdot \witsize \cdot \logF, \\
\size{\signrel} &= \tau \cdot (d-1)  \cdot \extdegVOLE \cdot \logF.
\end{align*}    
\end{property}
Similar to the TCitH framework, which requires $\nbatch  \cdot \extdeg \cdot \logF \geq \lambda$, the VOLEitH framework requires $\extdegVOLE \cdot \logF \geq \lambda$.
However, VOLEitH allows for a smaller $\para$, owing to its reduced soundness error, arguably at the cost of a more complex scheme.

\label{sec:TCitH}

\begin{table*}[tb]
    \centering
    \caption{Parameters of signatures based on \RSDP obtained via the TCitH framework. 
    The proof system parameters are selected analogously to those in \MQOM~\cite{NISTPQC-ADD-R2:MQOM24} to ensure comparable runtime performance of the VC scheme.}
    \label{tab:sizes_TCitH}
    
{\setlength{\tabcolsep}{4pt}
\begin{tabular}{
  c 
  S[table-format=3.0] 
  S[table-format=1.0] 
  S[table-format=4.0] 
  S[table-format=3.0] 
  S[table-format=2.0] 
  S[table-format=4.0]  
  S[table-format=1.0]  
  S[table-format=2.0]  
  S[table-format=2.0]  
  c 
  S[table-format=5.0] 
}

\toprule
\multirow{2}{*}{\begin{tabular}{c} Security\\ category \end{tabular}}
 &
 \multicolumn{4}{c}{\RSDP}
 & \multicolumn{5}{c}{Proof System}
 & \multicolumn{2}{c}{Performance}\\
 \cmidrule(lr){2-5}\cmidrule(lr){6-10} \cmidrule(lr){11-12}
 &
{$\cardinality{\F}$} &
{$\cardinality{\eset}$} &
{$\nC$} &
{$\kC$} &
{$\para$} &
{$\GGMsize$} &
{$\extdeg$} &
{$\nbatch$} &
{$\grindw$} &
{opt.} &
{$\size{\sign}$ in [\si{\byte}]} \\
\midrule
\multirow{4}{*}{NIST 1} 
%
& 
{\multirow{2}{*}{\begin{tabular}{S[table-format=3.0]} 127\end{tabular}}} & 
{\multirow{2}{*}{\begin{tabular}{S[table-format=1.0]}   7\end{tabular}}} & 
{\multirow{2}{*}{\begin{tabular}{S[table-format=4.0]} 127\end{tabular}}} & 
{\multirow{2}{*}{\begin{tabular}{S[table-format=3.0]}  76\end{tabular}}} & 
24 & 256 & 2 & 10 & 4 & fast & 7650 \\
&&&&&
15 & 2048 & 2 & 10 & 6 & short & 5533 \\
\cmidrule(lr){2-12}
%
&
{\multirow{2}{*}{\begin{tabular}{S[table-format=3.0]}   3\end{tabular}}} & 
{\multirow{2}{*}{\begin{tabular}{S[table-format=1.0]}   2\end{tabular}}} & 
{\multirow{2}{*}{\begin{tabular}{S[table-format=4.0]} 579\end{tabular}}} & 
{\multirow{2}{*}{\begin{tabular}{S[table-format=3.0]}  213\end{tabular}}} & 
17 & 256 & 6 & 14 & 9 & fast & 3533 \\
&&&&&
12 & 2048 & 7 & 12 & 8 & short & 3095 \\
\midrule
\multirow{4}{*}{NIST 3} 
%
& 
{\multirow{2}{*}{\begin{tabular}{S[table-format=3.0]} 127\end{tabular}}} & 
{\multirow{2}{*}{\begin{tabular}{S[table-format=1.0]}   7\end{tabular}}} & 
{\multirow{2}{*}{\begin{tabular}{S[table-format=4.0]} 187\end{tabular}}} & 
{\multirow{2}{*}{\begin{tabular}{S[table-format=3.0]}  111\end{tabular}}} & 
36 & 256 & 2 & 14 & 6 & fast & 16675 \\
&&&&&
23 & 2048 & 2 & 14 & 4 & short & 12354 \\
\cmidrule(lr){2-12}
%
&
{\multirow{2}{*}{\begin{tabular}{S[table-format=3.0]}   3\end{tabular}}} & 
{\multirow{2}{*}{\begin{tabular}{S[table-format=1.0]}   2\end{tabular}}} & 
{\multirow{2}{*}{\begin{tabular}{S[table-format=4.0]} 839\end{tabular}}} & 
{\multirow{2}{*}{\begin{tabular}{S[table-format=3.0]}  309\end{tabular}}} & 
27 & 256 & 6 & 21 & 3 & fast & 8284 \\
&&&&&
18 & 2048 & 7 & 18 & 12 & short & 6860 \\
\midrule
\multirow{4}{*}{NIST 5} 
%
& 
{\multirow{2}{*}{\begin{tabular}{S[table-format=3.0]} 127\end{tabular}}} & 
{\multirow{2}{*}{\begin{tabular}{S[table-format=1.0]}   7\end{tabular}}} & 
{\multirow{2}{*}{\begin{tabular}{S[table-format=4.0]} 251\end{tabular}}} & 
{\multirow{2}{*}{\begin{tabular}{S[table-format=3.0]}  150\end{tabular}}} & 
48 & 256 & 2 & 19 & 7 & fast & 29839 \\
&&&&&
31 & 2048 & 2 & 19 & 3 & short & 22305 \\
\cmidrule(lr){2-12}
%
&
{\multirow{2}{*}{\begin{tabular}{S[table-format=3.0]}   3\end{tabular}}} & 
{\multirow{2}{*}{\begin{tabular}{S[table-format=1.0]}   2\end{tabular}}} & 
{\multirow{2}{*}{\begin{tabular}{S[table-format=4.0]} 1102\end{tabular}}} & 
{\multirow{2}{*}{\begin{tabular}{S[table-format=3.0]}  406\end{tabular}}} & 
36 & 256 & 6 & 27 & 4 & fast & 14584 \\
&&&&&
25 & 2048 & 7 & 24 & 6 & short & 12608 \\
\bottomrule
\end{tabular}
}
    \vspace{-0.15cm}
\end{table*}

\section{Modeling Restricted Decoding Problems}\label{sec:crole}
We now present a polynomial modeling for \RSDP that can be utilized in the TCitH and VOLEitH frameworks. 
First, we recap some basic notions from coding theory.

\textbf{Coding theory.}
A \textit{linear code} over $\F$ of length $\nC$ is a $\kC$-dimensional subspace of $\F^\nC$. 
The code can be represented by a \textit{parity-check matrix} $\pcH \in \F^{(\nC-\kC) \times \nC}$ such that $\mathbf{c}\pcH^\top  = \60$ for every codeword $\mathbf{c}$. 
We denote by $\rC = \nC - \kC$ the redundancy of the code.
The \textit{syndrome} of a vector $\6v$ (through $\pcH$) is $\synd = \6v\pcH^\top \in \F^\rC$. 
The \textit{Hamming weight} (from now on, only \emph{weight}) of a vector $\6v$ is the number of non-zero entries in $\mathbf{v}$.

\textbf{Restricted decoding assumption.}
The (general) restricted Syndrome Decoding Problem (\RSDP) \cite{baldi2025new} replaces the Hamming-weight constraint of the usual SDP with the requirement that each entry must lie in the fixed subset $\eset$ of the underlying finite field $\F$.
It is defined as follows:

\begin{definition}[\RSDP]\label{RSDP}
Let the restriction $\eset \subset  \F$, the parity-check matrix $\pcH \in \F^{\rC\times \nC}$, and the syndrome $\synd \in \F^{\rC}$ be given. 
Solving \RSDP requires to find $\err \in \eset^\nC$ such that $\err \pcH^\top = \synd$.
\end{definition}
We refer to the size of the restriction as $\sizeE=\cardinality{\eset}$.
Different choices of $\eset$ are possible, and all of them lead to an NP-complete problem~\cite{PKC:BBPSWW24}.
In this work, we mainly focus on two specific choices that have undergone careful analysis.

\paragraph{CROSS restriction}
The signature scheme \CROSS~\cite{NISTPQC-ADD-R2:CROSS24} uses the restriction $\eset=\{2^i \mid i \in [7] \}$ within $\F$ of size $\fieldsize=127$.
We are going to refer to this setting as \CrossSDP, the concrete hardness of which is analyzed in detail in~\cite{cross:secguide} and confirmed in~\cite{beullens2024security}. 
For instance, the parameters $\nC=127$ and $\kC=76$ provide a suitable choice for NIST security category 1.

\paragraph{Ternary decoding with full weight}
The signature scheme \WAVE~\cite{AC:DebSenTil19,NISTPQC-ADD-R1:WAVE23} relies on the ternary syndrome decoding with high weight.
For state-of-the-art solvers~\cite{SAC:BCDL19}, parameterizing the problem such that the error is full-weight and unique, maximizes the concrete hardness (a setting in which \WAVE cannot sign successfully).
This corresponds to \RSDP with $\eset = \{1,2\}$ and $\fieldsize = 3$, to which we are going to refer as \TernarySDP.
When polynomial factors are conservatively ignored, the solver presented in~\cite{SAC:BCDL19} requires $2^{0.247\cdot n}$ operations for $\kC \approx 0.369\cdot \nC$, which implies uniqueness of the solution.
In particular, the parameters $\nC = 579$ and $\kC = 213$ are shown to achieve NIST security category 1. 

\paragraph{General \RSDP}
In general, any choice of restriction requires a careful analysis of the problem's concrete computational hardness.
This makes suggesting parameters difficult if one deviates from the settings described above.
Nevertheless, we are interested in assessing whether other parameterizations can yield an improved performance in combination with the proposed modeling.
An upper bound on the cost of solving a general \RSDP instance is provided by Stern's algorithm, as adapted in \cite{PKC:BBPSWW24}.
The following proposition provides a tight closed-form approximation for the cost of this solver.

\begin{proposition}\label{prop:stern}
Let $Z = \logtwo{\sizeE}$, $P=\logF$, and $k/n = 1-Z/P$. 
Then, conservatively ignoring polynomial factors and memory cost, Stern's solver has time complexity $C$ with 
\begin{equation}\label{eq:stern}
\logtwo{C} \approx \frac{P-Z}{2\cdot P - Z}\cdot Z\cdot \nC.
\end{equation}
\end{proposition}
\begin{proof}
For $k/n = 1-Z/P$, the solution $\err$ is unique with good probability~\cite{cross:secguide}.
Following~\cite{PKC:BBPSWW24}, the cost of Stern's solver is 
\[
\min_{0 \leq \ell \leq \rC} \left\{\sizeE^{\frac{\kC+\ell}{2}} + \frac{\sizeE^{\kC+\ell}}{\fieldsize^\ell} \right\},
\]
when polynomial factors and memory cost are neglected.
The optimal $\ell$ balances $\sizeE^{\frac{\kC+\ell}{2}} = \frac{\sizeE^{\kC+\ell}}{\fieldsize^\ell}$; this is achieved for $\kC/\ell = 2P/Z-1$.
\Cref{eq:stern} follows from $C \approx \sizeE^{\frac{\kC+\ell}{2}} = \fieldsize^\ell$.
\end{proof}
Progress in cryptanalysis of \RSDP would require increased $\nC$.
 As in the case of \cite{AC:BFGNR24}, since the signature size is typically dominated by $\signsym$ for the modeling we propose, this would have only a minor impact on performance in the TCitH and VOLEitH frameworks.

\textbf{Modeling \RSDP.}
The linearity of the syndrome equations allows a standard optimization that reduces the witness size from $\nC$ to $\kC$. 
We assume that the parity-check matrix is in systematic form
\footnote{This assumption does not impact security.} $\pcH = (\pcA, \Id{\rC})$, where $\pcA\in\F^{\rC\times \kC}$ and $\Id{\rC} \in \F^{\rC\times\rC}$ denotes the identity matrix.
Then, $\err \in \F^n$ satisfies $\err\pcH^\top = \synd$ if and only if it is of the form $\err = (\wit, \synd - \wit\pcA^\top)$ with $\wit\in\F^k$.
That is, the witness is defined as a length-$\kC$ partial error vector that is expanded to a length-$\nC$ error vector so that the syndrome equations are satisfied.

To check that the entries of the expanded witness satisfy the restriction\footnote{$\checkEOnly$ is also used to model \RSDP as a system of equations that can be attacked using algebraic solvers~\cite{beullens2024security,NISTPQC-ADD-R2:CROSS24}.} 
define $\checkEx \coloneqq \prod_{e\in\eset} (x-e)$.
Then, $\alpha \in \K$ is in $\eset$ if and only if $\checkE{\alpha} = 0$.

Overall, we obtain the system of  polynomial constraints $\RelPolyx = (f_1, \ldots, f_\nC) \in \F[x_1,\ldots,x_\kC]^\nC$ with
\begin{align}
f_i(\6x) &= \checkE{x_i}\quad \text{for }i \in [\kC], \notag\\
\text{and} \quad 
f_{\kC+i}(\6x) &= \checkE{s_i - \inprod{\arow_i}{\6x}}\quad \text{for }i \in [\rC], \label{eq:modelling_RSDP}
\end{align}
where $\arow_i \in \F^k$ denotes the $i$-th row of $\pcA$. 
The correctness and degree of this modeling are stated in the following proposition.

\begin{proposition}
The polynomial relation $\RelPolyOnly$ given in \Cref{eq:modelling_RSDP} provides a degree-$\sizeE$ modeling of \RSDP:
 $\wit \in \F^\kC$ satisfies $\RelPoly{\wit} = \60$, if and only if the error vector $\err = (\wit, \synd - \wit \pcA^\top)$ solves the \RSDP instance $(\pcH,\synd)$ with $\pcH = (\pcA, \Id{\rC})$.
\end{proposition}
\begin{proof}
For $i\in[\kC]$, $f_i(\6x)$ has degree $\sizeE$.
The composition of a degree-$\sizeE$ polynomial and a linear polynomial is also of degree $\sizeE$; and, thus, $f_{\kC+i}(\6x)$ and $\RelPolyx$ have degree $\sizeE$ as well.
Let $\wit$ satisfy $\RelPoly{\wit} = \60$.
Then, for $\err = (\wit, \synd - \wit \pcA^\top)$, it holds
\[
\err \pcH^\top = (\wit, \synd - \wit \pcA^\top) (\pcA, \Id{\rC})^\top = \wit \pcA^\top + \synd - \wit \pcA^\top = \synd,
\]
i.e., the syndrome equations are satisfied.
Further, $\checkE{e_i} = 0$ for all entries of $\err$, i.e., $\err\in\eset^\nC$. 
\end{proof}

\textbf{\RSDP-enabled tradeoffs.}
For the proposed modeling, the signature sizes in both TCitH and VOLEitH scale as $\size{\signrel} \propto \sizeE$ and $\size{\signwit} \propto \kC\cdot \logF$, while $\signsym$, composed of symmetric primitives, is largely independent of the modeling.
\Cref{fig:tradeoff} shows the tradeoff between $\kC\cdot\logF$ and $\sizeE$ which can be achieved by different parameterizations of \RSDP.
The solver cost is approximated by \Cref{prop:stern}.
This approximation is tight for both \TernarySDP and \CrossSDP, parametrized according to the careful analysis in \cite{SAC:BCDL19} and \cite{cross:secguide}, respectively.
We observe that only a minor improvement in witness size may be possible by increasing the degree of the relation (i.e., the size of $\eset$).
For fixed $\sizeE$, smaller field sizes $\fieldsize$ are preferable. 
In the following, we focus on \CrossSDP and \TernarySDP; by \Cref{fig:tradeoff}, the latter is expected to yield the smallest \RSDP-based signatures.

\begin{figure}
    \centering
    \begin{tikzpicture}[baseline] 
\begin{axis}[%
ymax = 1000,
xmin = 1.9,
xmax = 7.1,
label style={font=\footnotesize},
ymajorgrids,
xmajorgrids,
xminorgrids,
yminorgrids,
grid style=dashed,
legend style={at={(1.05,0.35)},anchor=west,legend cell align=left, align=left, draw=black, legend image post style={xscale=0.5,yscale = 0.5, mark size=5pt}, font=\scriptsize},
ylabel near ticks,
xlabel near ticks,
width = 0.7*\columnwidth,
height = 3.6cm,
ylabel={$\kC \cdot \logF$},
xlabel={degree $\sizeE$},
x tick label style={font=\scriptsize},
y tick label style={font=\scriptsize},
legend columns=1,
]

\addplot [only marks, CornflowerBlue, mark = triangle*] 
  table[row sep=crcr]{
2 338\\ 
};
\addlegendentry{$\TernarySDP$} 

\addplot [only marks, BrickRed, mark = otimes] 
  table[row sep=crcr]{
7 532 \\ 
};
\addlegendentry{$\CrossSDP$} 

\addplot [only marks, mark=x, thick] 
  table[row sep=crcr]{
2 310 \\ 
};
\addlegendentry{$\fieldsize = 3$} 

\addplot [thick] 
  table[row sep=crcr]{
2 521 \\ 
3 277 \\ 
4 189 \\ 
};
\addlegendentry{$\fieldsize = 5$} 

\addplot [dashed, thick] 
  table[row sep=crcr]{
2 1026 \\ 
3 593 \\ 
4 442 \\ 
5 360 \\ 
6 307 \\ 
7 274 \\ 
8 246 \\ 
9 225 \\ 
};
\addlegendentry{$\fieldsize = 17$} 



\addplot [dashdotted, thick] 
  table[row sep=crcr]{
2 1554 \\ 
3 926 \\ 
4 706 \\ 
5 588 \\ 
6 511 \\ 
7 457 \\ 
8 422 \\ 
9 392 \\ 
};
\addlegendentry{$\fieldsize = 61$} 

\addplot [densely dotted, thick] 
  table[row sep=crcr]{
2 1853 \\ 
3 1119 \\ 
4 853 \\ 
5 713 \\ 
6 629 \\ 
7 567 \\ 
8 525 \\ 
9 483 \\ 
};
\addlegendentry{$\fieldsize = 127$} 

  \draw[->, ForestGreen, thick] (axis cs:6.6,900) -- (axis cs:5.6,750)
    node[midway, sloped, above] {\scriptsize better};

\end{axis}
\end{tikzpicture}
    \caption{Tradeoff between degree and witness size. Black lines correspond to \RSDP parameters achieving NIST security category 1, as per \Cref{eq:stern}. \TernarySDP and \CrossSDP are parametrized according to \cite{SAC:BCDL19} and \cite{cross:secguide}.}
    \label{fig:tradeoff}
    \vspace{-3mm}
\end{figure}

\begin{table*}[tb]
    \centering
    \caption{Parameters of signatures based on \RSDP obtained via the VOLEitH framework. 
    The proof system parameters are selected analogously to those in \SDitH~\cite{NISTPQC-ADD-R2:SDitH24} to ensure comparable runtime performance of the VC scheme.}
    \label{tab:sizes_VOLEitH}
    
{\setlength{\tabcolsep}{4pt}
\begin{tabular}{
  c 
  S[table-format=3.0] 
  S[table-format=1.0] 
  S[table-format=4.0] 
  S[table-format=3.0] 
  S[table-format=2.0] 
  S[table-format=4.0]  
  S[table-format=3.0]  
  S[table-format=3.0]  
  S[table-format=2.0]  
  c 
  S[table-format=5.0] 
}

\toprule
\multirow{2}{*}{\begin{tabular}{c} Security\\ category \end{tabular}}
 &
 \multicolumn{4}{c}{\RSDP}
 & \multicolumn{5}{c}{Proof System}
 & \multicolumn{2}{c}{Performance}\\
 \cmidrule(lr){2-5}\cmidrule(lr){6-10} \cmidrule(lr){11-12}
 &
{$\cardinality{\F}$} &
{$\cardinality{\eset}$} &
{$\nC$} &
{$\kC$} &
{$\para$} &
{$\GGMsize$} &
{$\extdegVOLE$} &
{$\Topen$} &
{$\grindw$} &
{opt.} &
{$\size{\sign}$ in [\si{\byte}]} \\
\midrule
\multirow{4}{*}{NIST 1} 
%
& 
{\multirow{2}{*}{\begin{tabular}{S[table-format=3.0]} 127\end{tabular}}} & 
{\multirow{2}{*}{\begin{tabular}{S[table-format=1.0]}   7\end{tabular}}} & 
{\multirow{2}{*}{\begin{tabular}{S[table-format=4.0]} 127\end{tabular}}} & 
{\multirow{2}{*}{\begin{tabular}{S[table-format=3.0]}  76\end{tabular}}} & 
16 & 256 & 32 & 101 & 3 & fast & 6432 \\
&&&&&
11 & 2048 & 22 & 107 & 10 & short & 4372 \\
\cmidrule(lr){2-12}
%
&
{\multirow{2}{*}{\begin{tabular}{S[table-format=3.0]}   3\end{tabular}}} & 
{\multirow{2}{*}{\begin{tabular}{S[table-format=1.0]}   2\end{tabular}}} & 
{\multirow{2}{*}{\begin{tabular}{S[table-format=4.0]} 579\end{tabular}}} & 
{\multirow{2}{*}{\begin{tabular}{S[table-format=3.0]} 213\end{tabular}}} & 
16 & 256 & 96 & 101 & 1 & fast & 3515 \\
&&&&&
11 & 2048 & 81 & 107 & 8 & short & 2974 \\
\midrule
\multirow{4}{*}{NIST 3} 
%
& 
{\multirow{2}{*}{\begin{tabular}{S[table-format=3.0]} 127\end{tabular}}} & 
{\multirow{2}{*}{\begin{tabular}{S[table-format=1.0]}   7\end{tabular}}} & 
{\multirow{2}{*}{\begin{tabular}{S[table-format=4.0]} 187\end{tabular}}} & 
{\multirow{2}{*}{\begin{tabular}{S[table-format=3.0]}  111\end{tabular}}} & 
24 & 256 & 48 & 153 & 3 & fast & 14359 \\
&&&&&
16 & 4096 & 32 & 157 & 3 & short & 9361 \\
\cmidrule(lr){2-12}
%
&
{\multirow{2}{*}{\begin{tabular}{S[table-format=3.0]}   3\end{tabular}}} & 
{\multirow{2}{*}{\begin{tabular}{S[table-format=1.0]}   2\end{tabular}}} & 
{\multirow{2}{*}{\begin{tabular}{S[table-format=4.0]} 839\end{tabular}}} & 
{\multirow{2}{*}{\begin{tabular}{S[table-format=3.0]}  309\end{tabular}}} & 
24 & 256 & 144 & 153 & 1 & fast & 7816 \\
&&&&&
16 & 4096 & 128 & 157 & 1 & short & 6463 \\
\midrule
\multirow{4}{*}{NIST 5} 
%
& 
{\multirow{2}{*}{\begin{tabular}{S[table-format=3.0]} 127\end{tabular}}} & 
{\multirow{2}{*}{\begin{tabular}{S[table-format=1.0]}   7\end{tabular}}} & 
{\multirow{2}{*}{\begin{tabular}{S[table-format=4.0]} 251\end{tabular}}} & 
{\multirow{2}{*}{\begin{tabular}{S[table-format=3.0]}  150\end{tabular}}} & 
32 & 256 & 64 & 206 & 3 & fast & 25573 \\
&&&&&
21 & 4096 & 42 & 216 & 7 & short & 16589 \\
\cmidrule(lr){2-12}
%
&
{\multirow{2}{*}{\begin{tabular}{S[table-format=3.0]}   3\end{tabular}}} & 
{\multirow{2}{*}{\begin{tabular}{S[table-format=1.0]}   2\end{tabular}}} & 
{\multirow{2}{*}{\begin{tabular}{S[table-format=4.0]} 1102\end{tabular}}} & 
{\multirow{2}{*}{\begin{tabular}{S[table-format=3.0]}  406\end{tabular}}} & 
32 & 256 & 192 & 206 & 1 & fast & 13851 \\
&&&&&
21 & 4096 & 168 & 216 & 5 & short & 11521 \\
\bottomrule
\end{tabular}
}
    \vspace{-0.15cm}
\end{table*}

\section{Concrete Performance Evaluation}\label{sec:performance}

\textbf{Performance TCitH.}
From the provided modeling of \RSDP, a digital signature is derived using the TCitH framework~\cite{feneuil2025threshold} and the Fiat-Shamir Transform, as summarized in \Cref{sec:TCitH}.
The chosen parameters and the obtained signature sizes are provided in \Cref{tab:sizes_TCitH}.
The parameters $\para$, $\GGMsize$, and $\grindw$ of the proof system are selected following \cite{NISTPQC-ADD-R2:MQOM24} to allow for comparable runtime performance of the VC scheme.
To simplify implementation, the smallest extension degree $\extdeg = [\K:\F]$ satisfying the requirement $\GGMsize \leq \cardinality{\K} = \cardinality{\F}^\extdeg$ is chosen.
This determines the batching parameter $\nbatch$.

For \TernarySDP, parameterized to achieve NIST security category 1, a signature size as small as \SI{3.1}{\kilo\byte} is possible.
Switching to \CrossSDP, the signature size increases to \SI{5.5}{\kilo\byte}.
This is mainly due to the increased degree of the modeling and a larger witness size.
Nevertheless, we obtain a significant reduction compared to \CROSS~\cite{NISTPQC-ADD-R2:CROSS24}, for which the smallest signature size is \SI{12.4}{\kilo\byte}.

\textbf{Performance VOLEitH.}
\Cref{tab:sizes_VOLEitH} provides an overview of the performances that can be achieved if TCitH is replaced by the VOLEitH framework.
Parameters of the proof system are chosen analogously to \SDitH~\cite{NISTPQC-ADD-R2:SDitH24}.
For \TernarySDP, only a minor improvement in signature size is achievable.
A similar effect is observed in \cite{NISTPQC-ADD-R2:MQOM24}: for OWFs with small degree and small witness size, the structurally simpler TCitH framework remains competitive.
On the other hand, for \CrossSDP, which has a degree of $d = \sizeE = 7$, VOLEitH enables a reduction in signature size of about \SI{1}{\kilo\byte} for NIST category 1, and of up to \SI{6}{\kilo\byte} for NIST category 5.

\textbf{Comparison with NIST submissions.}
A comparison of the achieved signature sizes with other ZK-based signatures submissions to NIST's call for additional digital signature schemes~\cite{NISTcall} is provided in \Cref{tab:comparison}.
As can be seen, the TCitH and VOLEitH frameworks can be combined with a wide range of OWFs.
The proposed modeling of \RSDP is competitive compared to the state-of-the-art: for the \TernarySDP, we obtain a signature size which is very close to those of the smallest proposals that utilize the TCitH or the VOLEitH framework with an alternative hardness assumption.
The sizes obtained from \CrossSDP are slightly larger than those of \SDitH~\cite{NISTPQC-ADD-R2:SDitH24} and \PERK~\cite{NISTPQC-ADD-R2:RYDE24}, which also use modelings of degree greater than two, but smaller than $\sizeE = 7$.
The code-based proposals \CROSS and \LESS target other trade-offs:
\CROSS uses a simple CVE-like protocol \cite{CayVerElY10,NISTPQC-ADD-R2:CROSS24}, and \LESS achieves a small signature size by using multiple public keys (resulting in a public-key size that is multiple times larger than those of the other schemes shown)~\cite{NISTPQC-ADD-R2:LESS24}.
\SQISign~\cite{NISTPQC-ADD-R2:SQIsign24} is significantly more compact than the other schemes; this comes at the cost of a complex and slow construction.

\begin{table}[tb]
    \centering
    \caption{Comparison with other ZK-based signatures. Shown sizes are for the short optimization tradeoff; runtime may vary.}
    \label{tab:comparison}
    {\setlength{\tabcolsep}{3pt}
\begin{tabular}
{
ccc
S[table-format=2.2]
}
\toprule
{Scheme} 
&  {Assumption} 
& {Design} 
& {$\size{\sign}$ in [\si{\kilo\byte}]} \\
\midrule
\multirow{2}{*}{this work} & \CrossSDP  &  VOLEitH  & 4.4\\
 & \TernarySDP & TCitH & 3.1\\
\midrule
\cite{battagliola2025mpcith} &  \CROSS-\RSDP & MPCitH & 5.5 \\
\midrule
\CROSS &  \CROSS-\RSDP & CVE & 12.4 \\ 
\LESS
& Code Equivalence & sigma protocol & 1.3\\ 
\midrule
\SDitH & Syndrome Decoding & VOLEitH & 3.7 \\ 
\RYDE & Rank Decoding & TCitH & 3.1\\ 
\Mirath & MinRank & TCitH & 3.0 \\ 
\PERK & Permuted Kernel & VOLEitH & 3.5 \\ 
\MQOM & Multivariate Quadratic & TCitH & 2.8\\ 
\midrule
\multirow{2}{*}{\FAEST} & AES & \multirow{2}{*}{VOLEitH} & 4.5 \\ 
& AES-EM &  & 3.9\\ 
\midrule
\SQISign & {Supersingular Isogenies} & {sigma protocol} & 0.15\\
\bottomrule
\end{tabular}
\vspace{-2mm}
}
\end{table}

\textbf{Comparison with \cite{battagliola2025mpcith}.}
Recently  \cite{battagliola2025mpcith} proposed a new MPCitH-based protocol for \RSDP.
Instead of relying on the TCitH and VOLEitH frameworks, their construction is built from the ground up, designing an MPC protocol that computes the syndrome corresponding to a shared restricted vector and turning it into a digital signature using the classical MPCitH framework. 
The approach presented in this paper offers greater flexibility in parameter selection and achieves smaller signatures, in particular through \TernarySDP.

\section{Conclusion}

This paper presents a polynomial modeling of the restricted decoding problem \RSDP, enabling the construction of digital signatures in the TCitH and the VOLEitH frameworks.
When instantiated with ternary full-weight decoding, we obtain signature sizes that are competitive with the state-of-the-art NIST candidates.
For the restriction used by the \CROSS signature scheme, the sizes are significantly reduced compared to \CROSS.

\balance
\bibliographystyle{IEEEtran}
\bibliography{aux/abbrev3,aux/crypto_crossref,aux/ref}

\end{document}